\documentclass[11pt,letterpaper]{article}
\usepackage{authblk}
\usepackage{fullpage}
\usepackage[utf8]{inputenc}
\usepackage[english]{babel}
\usepackage{amsmath}
\usepackage{amsfonts}
\usepackage{amssymb}
\usepackage{makeidx}
\usepackage{graphicx}
\usepackage{lmodern}
\usepackage{tikz}
\usetikzlibrary{arrows}
\usetikzlibrary{positioning}
\usepackage{amsthm}
\usepackage[textsize=footnotesize]{todonotes}
\usepackage{bm}
\usepackage{enumitem}
\usepackage{color}

\usepackage{subcaption}
\usepackage{tkz-graph}

\usetikzlibrary{shapes.geometric}

\newtheorem{thrm}{Theorem}[section]

\newtheorem{defi}[thrm]{Definition}
\newtheorem{lem}[thrm]{Lemma}
\newtheorem{corollary}[thrm]{Corollary}
\newtheorem{prop}[thrm]{Proposition}

{\theoremstyle{definition}

\newtheorem{remark}[thrm]{Remark}}

\def\G{\mathcal{G}}
\def\S{\mathcal{S}}
\def\R{\mathcal{R}}
\def\V{\mathcal{V}}
\def\E{\mathcal{E}}

\def\sfa{\mathsf{a}}

\def\F{\mathbb{F}}

\def\N{\mathbb{N}}
\def\Z{\mathbb{Z}}

\DeclareMathOperator\mincut{mincut}
\DeclareMathOperator\sign{sign}

\newcommand{\net}[1]{$\triangledown\!_{#1}$}
\newcommand{\netm}[1]{\triangledown\!_{#1}}

\providecommand{\abs}[1]{\left\vert#1\right\vert}
\newcommand{\bigslant}[2]{{\raisebox{.1em}{$#1$}/\raisebox{-.1em}{$#2$}}}

\begin{document}

\title{Multicast Triangular Semilattice Network}
\author[1]{Angelina Grosso, Felice Manganiello, Shiwani Varal, and Emily Zhu}
\date{} 

\maketitle 

\abstract{We investigate the structure of the code graph of a multicast network that has a characteristic shape of an inverted equilateral triangle. We provide a criterion that determines the validity of a receiver placement within the code graph, present invariance properties of the determinants corresponding to receiver placements under symmetries, and provide a complete study of these networks' receivers and required field sizes up to a network of 4 sources. We also improve on various definitions related to code graphs.}

\section{Introduction}
A communication network is a collection of directed links connecting
transmitters, switches, and receivers, whose underlying structure can
be mathematically represented by a directed graph $\mathcal{G} =
(\mathcal{V},\mathcal{E})$ as introduced by Li \emph{et al.} \cite{ACLY00}. Koetter and M\'edard in \cite{KM03} studied the network code design as an algebraic problem that depends on the structure of the underlying graph. They made a connection between a given network information flow problem and an algebraic variety over the closure of a finite field.

In particular, a multicast network is an error-free network with unit-capacity channels represented by a directed acyclic graph and with the communication requirement that every receiver demands the message sent by every source. Treating the messages as elements of some large enough finite field $\F_q$, it is known that linear network coding suffices to transmit the maximal number of messages.

Code graphs condense the information in a choice of edge-disjoint paths of a multicast network based on the coding points, i.e.\ edges which are ``bottlenecks" where messages are combined in linear network coding. Under this framework, linear network coding is reduced to assigning vectors to vertices in the code graph with independence conditions based on receivers. The triangular semilattice networks are then a family of code graphs embedded in the integer lattice restricted to nonnegative coordinates of some maximum 1-norm with edges between adjacent lattice points directed towards the origin.

This paper is organized as follows. In Section $2$, we refer briefly to and improve upon coding points, code graph, and $\mathbb{F}_q$-labeling of a code graph; these are discussed in detail in \cite{anderson}. We then present a result with regards to determinants in the $\F_q$-labeling. In Section $3$, we introduce a type of code graph called the triangular semilattice network. We discuss receiver placements and invariance of the minors corresponding to receiver placements under symmetries. From this general study, we shift to a complete study of triangular semilattice network with up to four sources in section $4$.

\section{Coding Points and Code Graph}

In this work, we represent a multicast network by a
directed acyclic graph $\G=(\V,\E)$ with a set $\S\subset \V$ of
sources, i.e.\ vertices without incoming edges, and a set
$\R\subset \V$ of receivers, i.e.\ vertices without outgoing edges. Each
directed edge is a unit capacity noise-free communication channel over
a finite field $\F_q$. We further assume that the edge mincut between each
source and each receiver is at least one and the overall mincut
between the set of sources and each receiver is at least the number of
sources. Together with the assumption of coordination at source level
and with the requirement that every receiver $R \in \R$ gets the message from
every source $S \in \S$, the network is equivalent to a multicast network as
defined in \cite{ACLY00}.

If $\R$ consists of a single receiver, the communication requirement is satisfied by a routing solution if and only if $\abs{\S} \leq \mincut(\S,R)$ as a result of Menger's Theorem, which states that the edge $\mincut(\S,R)$ is equal to the maximum number of edge-disjoint paths between the source set $\S$ and the receiver $R$ \cite{anderson}. In the case of multiple receivers where $\abs{\S} \leq \min_{R \in \mathcal{R}} \mincut(\S,R)$, Ahlswede \emph{et al.} \cite{ACLY00} first showed that a network coding solution exists; later it was found that a linear network coding solution over a finite field $\F_q$ exists when $q$ is sufficiently large \cite{LYC03}, in particular, $q \geq \abs{\R}$ was found to be sufficient \cite{jaggi05}. Interested readers may also refer to \cite{Ksc11} for a complete algebraic proof showing that $q > \abs{\R}$ is sufficient.

To condense the information about these
receiver requirements, we consider the corresponding code graph of a
multicast network. Anderson \emph{et al.} \cite{anderson} explain coding points of a network as the bottlenecks of the network where the linear combinations occur. More formally:
\begin{defi}
Let $\mathcal{G}$ be the underlying directed acyclic graph of a multicast network and for each $R \in \mathcal{R}$ let $\mathcal{P}_R = \{P_{S,R} \mid S \in \mathcal{S}\}$ be a set of edge-disjoint paths, where $P_{S,R}$ denotes a path from $S$ to $R$. A coding point of $\mathcal{G}$ is an edge $e=(v,v') \in \mathcal{E}$ such that:
\begin{itemize}
\item There are distinct sources $S, S'\in \S$ and distinct receivers
  $R,R'\in \R$ such that $e$ appears in both $P_{S,R} \in \mathcal{P}_R$ and $P_{S',R'} \in \mathcal{P}_{R'}$.
\item The parents of $v$ in $P_{S,R}$ and $P_{S',R'}$ are distinct.
\end{itemize}
\end{defi}
\begin{defi}
A coding-direct path in $\G$ from $v_1 \in \mathcal{V}$ to $v_2 \in \mathcal{V}$ is a path from $v_1$ to $v_2$ that does not pass through any coding point in $\G$, except possibly in the first edge.
\end{defi}

Note that coding points are dependent on the choices of edge disjoint paths to each receiver. With
$\mathcal{G}=(\mathcal{V},\E,\S,\R ,\{\mathcal{P}_R\mid R \in
\mathcal{R}\})$ we denote a multicast network with chosen sets
of edge-disjoint paths from the sources to each receiver. For a given
multicast network, Anderson \emph{et al.} \cite{anderson} define the
code graph as a directed graph with labeled vertices that preserves
the essential information of the network:
\begin{defi}
Let $\mathcal{G}=(\mathcal{V},\E,\S,\R ,\{\mathcal{P}_R\mid R \in
\mathcal{R}\})$ be a multicast network and let $\mathcal{Q}$ be its
set of coding points. Let the code graph $\Gamma = \Gamma(\G)$ be the
vertex-labeled directed acyclic graph constructed as follows:  
\begin{itemize}
\item The vertex set of $\Gamma$ is $\mathcal{S} \cup \mathcal{Q}$. Given a vertex $v$ of $\Gamma$, the corresponding source or coding point in $\G$ is called the $\G$-object of $v$.
\item The edge set of $\Gamma$ is the set of all ordered pairs of vertices of $\Gamma$ such that there is a coding-direct path in $\G$ between the corresponding $\G$-objects.
\item Each vertex $v$ of $\Gamma$ is labeled with a subset $L_v\subseteq\mathcal{R}$. A receiver $R\in \mathcal{R}$ is in $L_v$ if and only if there is a coding-direct path in $\G$ from the $\G$-object of $v$ to $R$.
\end{itemize}
\end{defi}
In general, Anderson \emph{et al.} \cite{anderson} present the following proposition that attempts to outline the properties of a code graph:

\begin{prop} 
	For any code graph $\Gamma=\Gamma(\G)$, we have that:
	\begin{itemize}
	\item $\Gamma$ is an acyclic graph.
	\item every vertex in $\Gamma$ either has in-degree 0, in which case its $\G$-object is a source, or it has in-degree at least 2, in which case its $\G$-object is a coding point.
	\item for each $R\in\R$, the set of vertices $V_R = \{v \in V \mid R \in L_v\}$ has cardinality $\abs{\S}$, and there are $\abs{\S}$ vertex-disjoint paths from the sources to this set corresponding to the original $\abs{\S}$ edge-disjoint paths.
	\end{itemize}
\end{prop}

The networks we consider in this work will satisfy these
properties. Nonetheless, the condition on the in-degree of
a coding point seems to require additional constraints. In
Figure~\ref{fig:badcodepath}, the code graph construction only
produces one edge to the bottom coding point.

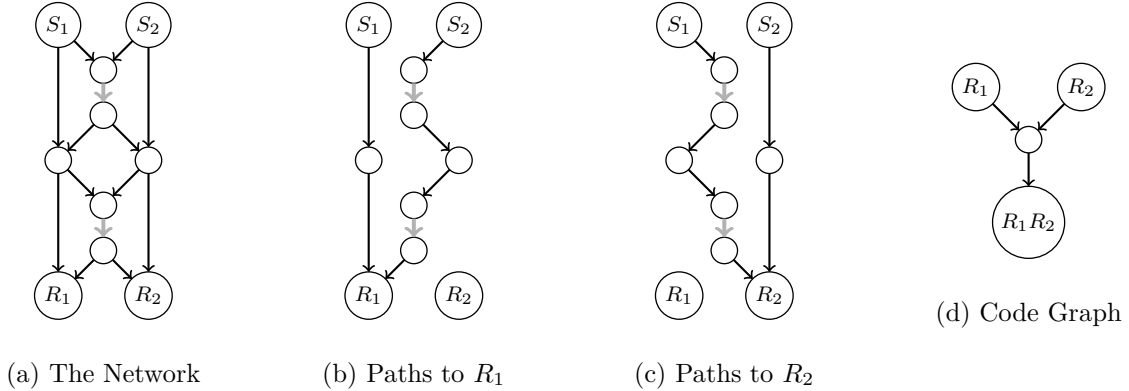
\begin{figure}[htbp]
\begin{subfigure}{.25\textwidth}
\begin{center}
\begin{tikzpicture}
\SetGraphUnit{.6}
\SetVertexMath
\tikzstyle{every node}=[font=\scriptsize]
\renewcommand*{\VertexInterMinSize}{10pt}
\Vertex{S_1} \SOEA[NoLabel](S_1){A} \NOEA(A){S_2}
\SetVertexNoLabel
\SO(A){B} \SOWE(B){C} \SOEA(B){D} \SOEA(C){E} \SO(E){F}
\SetVertexLabel
\SOWE(F){R_1} \SOEA(F){R_2}
\SetUpEdge[style={->}]
\Edge(S_1)(A) \Edge(S_2)(A) \Edge[color=black!30,lw=1.5pt](A)(B) \Edge(B)(C) \Edge(B)(D) \Edge(D)(E) \Edge(C)(E)
\Edge[color=black!30,lw=1.5pt](E)(F) \Edge(F)(R_1) \Edge(F)(R_2) \Edge(S_1)(C) \Edge(C)(R_1) \Edge(S_2)(D) \Edge(D)(R_2)
\end{tikzpicture}
\end{center}
\caption{The Network}
\end{subfigure}\begin{subfigure}{.25\textwidth}
\begin{center}
\begin{tikzpicture}
\SetVertexMath
\SetGraphUnit{.6}
\tikzstyle{every node}=[font=\scriptsize]
\renewcommand*{\VertexInterMinSize}{10pt}
\Vertex{S_1} \SOEA[NoLabel](S_1){A} \NOEA(A){S_2}
\SetVertexNoLabel
\SO(A){B} \SOWE(B){C} \SOEA(B){D} \SOEA(C){E} \SO(E){F}
\SetVertexLabel
\SOWE(F){R_1} \SOEA(F){R_2}
\SetUpEdge[style={->}]
\Edge(S_2)(A) \Edge[color=black!30,lw=1.5pt](A)(B) \Edge(B)(D) \Edge(D)(E)
\Edge[color=black!30,lw=1.5pt](E)(F) \Edge(F)(R_1) \Edge(S_1)(C) \Edge(C)(R_1)
\end{tikzpicture}
\end{center}
\caption{Paths to $R_1$}
\end{subfigure}\begin{subfigure}{.25\textwidth}
\begin{center}
\begin{tikzpicture}
\SetVertexMath
\SetGraphUnit{.6}
\tikzstyle{every node}=[font=\scriptsize]
\renewcommand*{\VertexInterMinSize}{10pt}
\Vertex{S_1} \SOEA[NoLabel](S_1){A} \NOEA(A){S_2}
\SetVertexNoLabel
\SO(A){B} \SOWE(B){C} \SOEA(B){D} \SOEA(C){E} \SO(E){F}
\SetVertexLabel
\SOWE(F){R_1} \SOEA(F){R_2}
\SetUpEdge[style={->}]
\Edge(S_1)(A) \Edge[color=black!30,lw=1.5pt](A)(B) \Edge(B)(C) \Edge(C)(E)
\Edge[color=black!30,lw=1.5pt](E)(F) \Edge(F)(R_2) \Edge(S_2)(D) \Edge(D)(R_2)
\end{tikzpicture}
\end{center}
\caption{Paths to $R_2$}
\end{subfigure}\begin{subfigure}{.24\textwidth}\begin{center}
\begin{tikzpicture}
\SetVertexMath
\tikzstyle{every node}=[font=\scriptsize]
\renewcommand*{\VertexInterMinSize}{10pt}
\SetGraphUnit{.7}
\Vertex{R_1} \SOEA[NoLabel](R_1){A} \NOEA(A){R_2} \SetGraphUnit{1.1}\SO(A){R_1R_2} 
\SetUpEdge[style={->}]
\Edge(R_1)(A) \Edge(R_2)(A) \Edge(A)(R_1R_2)
\end{tikzpicture}
\end{center}
\caption{Code Graph}
\end{subfigure}
\caption{Convoluted choice of paths}
\label{fig:badcodepath}
\end{figure}

A slight modification of this construction shows that taking a set of paths with the minimum number of coding points is insufficient to guarantee that the in-degree of every coding point is at least two. For simplicity, edges between sources and receivers are omitted.
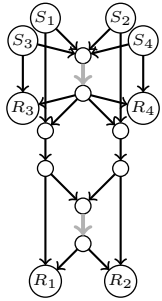
\begin{figure}[htbp]
\begin{center}
\begin{tikzpicture}
\SetGraphUnit{.5}
\SetVertexMath
\tikzstyle{every node}=[font=\tiny]
\renewcommand*{\VertexInterMinSize}{6pt}
\renewcommand*{\VertexInnerSep}{.5pt}
\Vertex{S_1} \SOEA[NoLabel](S_1){A} \NOEA(A){S_2}
\SetVertexNoLabel
\SO(A){B} \SOWE(B){C} \SOEA(B){D} \SO(C){G} \SO(D){H} \SOEA(G){E} \SO(E){F}
\SetVertexLabel
\SOWE(F){R_1} \SOEA(F){R_2} 
\SetGraphUnit{.3} \SOWE(S_1){S_3} \SOEA(S_2){S_4} \NOWE(C){R_3} \NOEA(D){R_4}
\SetUpEdge[style={->}]
\Edge(S_1)(A) \Edge(S_2)(A) \Edge[color=black!30,lw=1.5pt](A)(B) \Edge(B)(C) \Edge(B)(D) \Edge(C)(G) \Edge(D)(H) \Edge(G)(E) \Edge(H)(E)
\Edge[color=black!30,lw=1.5pt](E)(F) \Edge(F)(R_1) \Edge(F)(R_2) \Edge(S_1)(C) \Edge(G)(R_1) \Edge(S_2)(D) \Edge(H)(R_2)
\Edge(S_3)(A) \Edge(S_4)(A) \Edge(B)(R_3) \Edge(B)(R_4) \Edge(S_3)(R_3) \Edge(S_4)(R_4)
\end{tikzpicture}
\end{center}
\caption{Bottom coding point has in-degree one when taking paths analogous to the above}
\end{figure}

Anderson \emph{et al.} \cite{anderson} also provide a criterion to determine when a labeled network is a code graph:
\begin{prop}
Let $\Gamma = (V,E)$ be a vertex-labeled, directed acyclic graph where each vertex $v$ is labeled with a finite set $L_v$. Let $\mathcal{S} := \{v \in V \mid v \text{ has in-degree } 0\}, \mathcal{Q} := V\backslash\S$, and $\mathcal{R} = \bigcup_{v\in V} L_v$. Suppose:
\begin{itemize}
\item The in-degree of every vertex in $\mathcal{Q}$ is at least 2.
\item For each $R\in\mathcal{R}$, the set $V_R = \{v \in V: R \in L_v\}$ has $\abs{\mathcal{S}}$ vertices.
\item For each $R \in \mathcal{R}$ there is a set $\Pi_R = \{\pi_{S,R}\mid S\in \mathcal{S}\}$ of vertex-disjoint paths where every vertex and edge of $\Gamma$ is contained in some $\pi_{S,R}$.b
\end{itemize}
Then $\Gamma$ is the code graph for a reduced multicast network whose sources, coding points, and receivers are in one-to-one correspondence with the elements of $\mathcal{S},\mathcal{Q},$ and $\mathcal{R}$, respectively.
\end{prop}

In Figure~\ref{fig:badpath}, we find that the condition that a single choice of vertex-disjoint paths using all edges and vertices may be insufficient to guarantee that a graph is a code graph of some multicast network. In this case, the bottom node cannot act as a coding point as the two paths to it originate from the same source. One can note that the edge between the coding points can be avoided completely when instead taking the path directly from the second source to the bottom coding point as the path to $R_1$. 
\begin{figure}[htbp]
\begin{subfigure}{.33\textwidth}
\begin{center}
\begin{tikzpicture}
\SetVertexMath
\tikzstyle{every node}=[font=\footnotesize]
\renewcommand*{\VertexInterMinSize}{15pt}
\SetGraphUnit{.8}
\Vertex{R_1} \SOEA(R_1){R_2} \NOEA[NoLabel](R_2){A} \SOEA(R_2){R_1R_2}
\SetUpEdge[style={->}]
\Edge(R_1)(R_2) \Edge(A)(R_2) \Edge(R_2)(R_1R_2) \Edge(A)(R_1R_2)
\end{tikzpicture}
\end{center}
\subcaption{The Code Graph}
\end{subfigure}\begin{subfigure}{.33\textwidth}
\begin{center}
\begin{tikzpicture}
\SetVertexMath
\tikzstyle{every node}=[font=\footnotesize]
\renewcommand*{\VertexInterMinSize}{15pt}
\SetGraphUnit{.8}
\Vertex{R_1} \SOEA(R_1){R_2} \NOEA[NoLabel](R_2){A} \SOEA(R_2){R_1R_2}
\SetUpEdge[style={->}]
 \Edge(A)(R_2) \Edge(R_2)(R_1R_2)
\end{tikzpicture}
\subcaption{Paths to $R_1$}
\end{center}
\end{subfigure}\begin{subfigure}{.33\textwidth}
\begin{center}
\begin{tikzpicture}
\SetVertexMath
\tikzstyle{every node}=[font=\footnotesize]
\renewcommand*{\VertexInterMinSize}{15pt}
\SetGraphUnit{.8}
\Vertex{R_1} \SOEA(R_1){R_2} \NOEA[NoLabel](R_2){A} \SOEA(R_2){R_1R_2}
\SetUpEdge[style={->}]
\Edge(R_1)(R_2)  \Edge(A)(R_1R_2)
\end{tikzpicture}
\end{center}
\subcaption{Paths to $R_2$}
\end{subfigure}
\caption{Convoluted Choice of Paths}
\label{fig:badpath}
\end{figure}
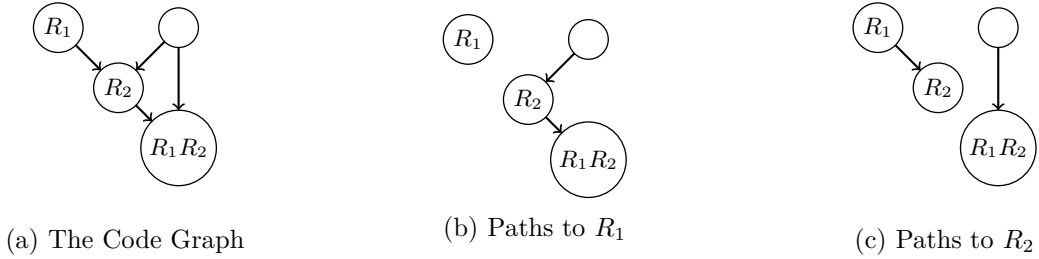

Note that it is still insufficient to require that all choices of vertex disjoint paths $\{\Pi_R\}_{R\in\R}$ use all edges/vertices. Consider Figure~\ref{fig:allpaths} below, which has the following forced vertex disjoint paths but for which the bottom vertex cannot be a coding point. Further in this paper, we will require various receiver placements which will ensure that the formed labeled directed acyclic graphs are code graphs.
\begin{figure}[htbp]
\begin{subfigure}{.33\textwidth}
\begin{center}
\begin{tikzpicture}
\tikzstyle{every node}=[font=\tiny]
\renewcommand*{\VertexInterMinSize}{15pt}
\SetGraphUnit{.65}
\Vertex[L={$R_2$}]{1} \SOEA[L={$R_1$}](1){4} \SOEA[L={$R_1R_2$}](4){6} \NOEA[NoLabel](4){2} \SOEA[L={$R_2$}](2){5} \NOEA[L={$R_1$}](5){3}

\SetUpEdge[style={->}]
\Edges(1,4,6) \Edges(2,4) \Edges(2,5) \Edges(3,5,6)
\end{tikzpicture}
\end{center}
\caption{The Code Graph}
\end{subfigure}\begin{subfigure}{.33\textwidth}
\begin{center}
\begin{tikzpicture}
\tikzstyle{every node}=[font=\tiny]
\renewcommand*{\VertexInterMinSize}{15pt}
\SetGraphUnit{.65}
\Vertex[L={$R_2$}]{1} \SOEA[L={$R_1$}](1){4} \SOEA[L={$R_1R_2$}](4){6} \NOEA[NoLabel](4){2} \SOEA[L={$R_2$}](2){5} \NOEA[L={$R_1$}](5){3}

\SetUpEdge[style={->}]
\Edges(1,4) \Edges(2,5,6)
\end{tikzpicture}
\end{center}
\caption{Paths to $R_1$}
\end{subfigure}\begin{subfigure}{.33\textwidth}
\begin{center}
\begin{tikzpicture}
\tikzstyle{every node}=[font=\tiny]
\renewcommand*{\VertexInterMinSize}{15pt}
\SetGraphUnit{.65}
\Vertex[L={$R_2$}]{1} \SOEA[L={$R_1$}](1){4} \SOEA[L={$R_1R_2$}](4){6} \NOEA[NoLabel](4){2} \SOEA[L={$R_2$}](2){5} \NOEA[L={$R_1$}](5){3}

\SetUpEdge[style={->}]
\Edges(2,4,6) \Edges(3,5)
\end{tikzpicture}
\end{center}
\caption{Paths to $R_2$}
\end{subfigure}
\caption{Only one choice of paths (but not a code graph)} 
\label{fig:allpaths}
\end{figure}
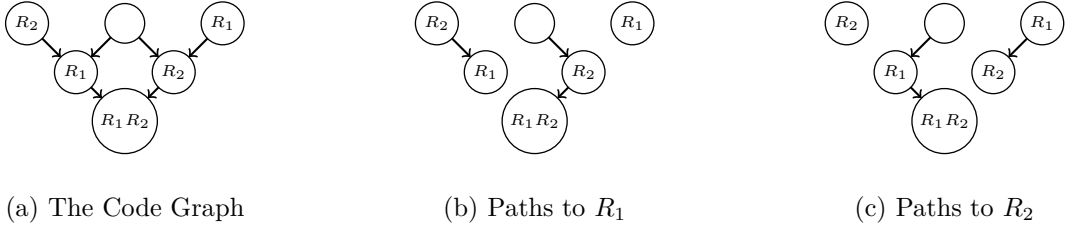

There exists extensive literature, eg. \cite{KM03}, \cite{Ksc11}, \cite{SYLL15}
, that follow the approach of assigning edge transfer coefficients or vertex transfer matrices directly to the multicast network. Fragouli and Soljanin \cite{FS06} introduced (as coding vectors) and Anderson \emph{et al.}\ \cite{anderson} expanded on the concept of $\mathbb{F}_q$-labelings of code graphs, which allow us to focus on the linear dependence and independence conditions of a single matrix. 

\begin{defi}
Let $\mathcal{G}=(\mathcal{V},\E,\S,\R ,\{\mathcal{P}_R\mid R \in
\mathcal{R}\})$ be a multicast network and $\Gamma=(V,E)$ be its corresponding code graph. Each $v \in V$ is labeled with a set of receivers $L_v \subseteq \mathcal{R}$. Let $V_R = \{v \in V \mid R \in L_v\}$. An $\mathbb{F}_q$-labeling of $\Gamma$ is an assignment of elements of $\mathbb{F}_q^{\abs{\mathcal{S}}}$ to the vertices of $\Gamma$ satisfying:
\begin{itemize}
\item The vectors assigned to the source nodes of the code graph are
  linearly independent and without loss of generality they can be chosen to be the standard basis.
\item The vectors assigned to vertices labeled with a common receiver are linearly independent.
\item The vector assigned to a coding point $Q \in V$ is in the span of vectors assigned to the tails of the directed edges terminating at $Q$.
\end{itemize}
We call the $\abs{\mathcal{S}} \times \abs{V}$ matrix consisting of the vectors of the $\F_q$-labeling, an $\F_q$-labeling matrix of $\Gamma$.
\end{defi}

Anderson \emph{et al.} \cite{anderson} note that the capacity of
$\mathcal{G}$ is achievable over $\mathbb{F}_q$ if and only if there
exists an $\mathbb{F}_q$-labeling of $\Gamma$. With this, it suffices
to examine properties of code graphs as opposed to complete
networks. In this paper, we study the solvability of a multicast
network over various finite fields upon the addition of receiver placements.

\begin{defi}
  \label{def:placementlabel}
  Let
  $\mathcal{G}=(\mathcal{V},\E,\S,\R ,\{\mathcal{P}_R\mid R \in
  \mathcal{R}\})$ be a multicast network and $\Gamma=(V,E)$ be its
  corresponding code graph and $R\in \R$. We call the set
  $V_R = \{v\in V \mid R \in L_v\}$ a receiver placement of $R$ and a
  vertex $v \in V_R$ a label of $R$ or more generally, a receiver label. The determinant of a receiver
  placement of $R$ is the maximal minor of the $\F_q$-labeling matrix
  of $\Gamma$ with columns corresponding to its
  labels.
\end{defi}

Since a set of vectors forming a square matrix is linearly independent if and only if the matrix's determinant is nonzero, we examine the structure of the determinants of receiver placements. In particular, to assist in determining if such an $\F_q$-labeling matrix exists, we will consider the matrix over $\F_q[\alpha_{(u,v)}: (u,v) \in E]$ formed by assigning the standard basis to the sources and variable linear combinations of the parents' vectors, i.e.\ if $N_u$ is the vector in the $\F_q$-labeling matrix corresponding to a vertex $u \in V$, for some $v \in \mathcal{Q}$, we would consider the vector $\sum_{u: (u,v) \in E} \alpha_{(u,v)} \cdot N_u$.
\begin{defi}
Let $\S=\{S_1,\dots,S_n\}$ and $V_R$ be a receiver placement,
i.e. $V_R = \{R^{(1)},\dots,R^{(n)}\}\subset V$. We introduce the following notations:
\begin{itemize}
\item $\pi_{i,j}$ denotes a path from $S_i$ to $R^{(j)}$.
\item $\Pi_{R,\sigma} = \{\pi_{i,\sigma(i)}\mid i \in [n]\}$ for some
  $\sigma \in \mathcal{S}_n$, where $[n] = \{i\}_{i=1}^n$ and
  $\mathcal{S}_n$ is the symmetric group of degree $n$, is a set of
  paths matching the sources to the receiver labeled vertices
\item $\Psi_R = \{\Pi_{R,\sigma}^{(j)} \mid \sigma \in \mathcal{S}_n,
  j \in [m_\sigma]\}$, where $m_\sigma$ is the number of paths,
  possibly 0, for this given matching of sources to receiver labeled
  vertices, consists of all sets of paths from the sources to the
  receiver labeled vertices.
\item $\Phi_R = \{\Pi_{R,\sigma} = \Pi_{R,\sigma}^{(j)} \in \Psi_R
  \mid j \in [m_\sigma], \pi_{i,\sigma(i)}$ are vertex disjoint$\}$
  consists of all sets of vertex disjoint paths from the sources to
  the receiver labeled vertices.

\end{itemize}
\end{defi}
Note that the $\sigma$ corresponding to $\Pi_{R,\sigma}$ is well-defined and unique as we have $n$ sources and $n$ labels, but for a given $\sigma$, $\Pi_{R,\sigma}$ is not necessarily unique---it may not even exist. In a slight abuse of notation, we will also write $(u,v) \in \Pi_{R,\sigma}$ to denote that $(u,v) \in \pi_{i,\sigma(i)}$ for some $\pi_{i,\sigma(i)} \in \Pi_{R,\sigma}$.

\begin{prop}
\label{prop:detterms}
Let $S_1,\dots,S_n$ denote the sources in a code graph with the $\F_q$-labeling matrix denoted $N$. Given a receiver placement of $R$, i.e.\ $V_R = \{R^{(1)},\dots,R^{(n)}\}$, we have
\[\det(N_R) = \sum_{\Pi_{R,\sigma} \in \Phi_R} \sign(\sigma) \prod_{(u,v) \in \Pi_{R,\sigma}} \alpha_{(u,v)}\in \F_q[\alpha_{(u,v)}: (u,v)\in E]\]
where $N_R$ is the submatrix of $N$ corresponding to
$R^{(1)},\dots,R^{(n)}$ and $\alpha_{(u,v)}$ is the transfer
coefficient, also called channel gain, corresponding to the edge
$(u,v)$ and $\F_q[\alpha_{(u,v)}: (u,v)\in E]$ is the multivariate polynomial ring
where variables correspond to the transfer coefficients. 
\end{prop}

This proposition says that the minor corresponding to a receiver
placement in a $\F_q$-labeling matrix can be calculated by the sum
over the sets of vertex disjoint paths to the receiver labeled vertices of the product of the transfer coefficients corresponding to the edges in any of those paths. In other words, sets including vertex-intersecting paths do not affect the minor.

We first show the following property about the set $\Psi_R \backslash
\Phi_R$ of sets of paths with vertex-intersecting paths.  
\begin{lem}
\label{lem:detmatch}
There is a matching of $\Psi_R \backslash \Phi_R$ without fixed points,
meaning a bijective map
$\mu:\Psi_R \backslash \Phi_R \rightarrow \Psi_R \backslash \Phi_R$
with $\mu \circ \mu = id$ and $\mu(\Pi_{R,\sigma})\neq\Pi_{R,\sigma}$ for all
$\Pi_R\in \Psi_R \backslash \Phi_R$, such that for $\mu(\Pi_{R,\sigma}) = \Pi'_{R,\sigma'}$:
\[\sign(\sigma) = -\sign(\sigma') \text{ and }\prod_{(u,v) \in \Pi_{R,\sigma}} \alpha_{(u,v)} = \prod_{(u,v) \in \mu(\Pi_{R,\sigma})} \alpha_{(u,v)}\]
\end{lem}
\begin{proof}
Let $\Pi_{R,\sigma} \in \Psi_R \backslash \Phi_R$ be arbitrary and let sources $S_{i}, S_{j}$ be the minimum $(i,j)$ (under lexicographic ordering) such that $\pi_{i,\sigma(i)}$ and $\pi_{j,\sigma(j)}$ intersect at some vertex. Let $x$ be the first vertex at which these paths intersect. Furthermore, let $\pi_{l,x} \subseteq \pi_{l,\sigma(l)}$ denote the subset of the path $\pi_{l,\sigma(l)}$ going from $S_{l}$ to $x$ and $\pi_{x,\sigma(l)} \subseteq \pi_{l,\sigma(l)}$ denote the subset of the path $\pi_{l,\sigma(l)}$ going from $x$ to $R^{(\sigma(l))}$ for $l = i,j$. 

We define $\mu(\Pi_{R,\sigma}) = \{\pi'_{k,\sigma'(k)}: k \in [n]\}$ where:
\[\sigma'(k) = \begin{cases}\sigma(k) &\text{if }k \neq i, j\\\sigma(j) &\text{if }k = i\\\sigma(i) &\text{if }k = j\end{cases} \text{ and } \pi'_{k,\sigma'(k)} = \begin{cases}\pi_{k,\sigma(k)} &\text{if }k \neq i,j\\\pi_{i,x} \cup \pi_{x,\sigma(j)} &\text{if }k=i\\\pi_{j,x} \cup \pi_{x,\sigma(i)} &\text{if }k=j\end{cases}\]
Note that this $\mu$ satisfies the desired properties:
\begin{itemize}
\item Clearly there is no $\mu(\Pi_{R,\sigma}) = \Pi_{R,\sigma}$ since necessarily distinct portions of the paths from two sources are swapped to get $\mu(\Pi_{R,\sigma})$.
\item $\mu\circ\mu(\Pi_{R,\sigma}) = \Pi_{R,\sigma}$ as the minimum $(i,j)$ and first vertex of intersection are the same for $\Pi_{R,\sigma}$ and $\mu(\Pi_{R,\sigma})$, so applying $\mu$ again simply swaps the swapped portion back the the original paths, returning $\mu(\mu(\Pi_{R,\sigma}))$ to $\Pi_{R,\sigma}$.
\item This is bijective since by the above, $\mu$ is its own inverse.
\item We have that $\sign(\sigma) = -\sign(\sigma')$ as $\sigma' = \tau_{i,j}\circ\sigma$ (where $\tau_{i,j}$ denotes the transposition of $i,j$, which fixes all other elements).
\item $\prod_{(u,v) \in \Pi_{R,\sigma}} \alpha_{(u,v)} = \prod_{(u,v) \in \mu(\Pi_{R,\sigma})} \alpha_{(u,v)}$ as both sets of paths use exactly the same edges with the same multiplicity by definition.\qedhere
\end{itemize}
\end{proof}

We now turn to the proof of the proposition:
\begin{proof}
{(Proposition \ref{prop:detterms})}
Note that by definition of determinant:
\[\det(N_R) = \sum_{\rho \in \mathcal{S}_n} \sign(\rho) \prod_{i=1}^n (N_R)_{i,\rho(i)}\]
where we note that $\rho(i)$ determines at which receiver a path ends and $i$ determines from which source a path originates. As such, based on the line graph (like in Kschischang's argument in Appendix C \cite{Ksc11}), we see that an entry of the matrix is the sum over the paths from $S_i$ to $R^{(\rho(i))}$ of the product over the edges of the transfer coefficients, so:
\[(N_R)_{i,\rho(i)} = \sum_{\pi_{i,\rho(i)}: \text{ a path}}\  \prod_{(u,v) \in \pi_{i,\rho(i)}} \alpha_{(u,v)}\]
where $\pi_{i,\rho(i)}$ is any path from $S_i$ to $R^{(\rho(i))}$. Now expanding $\prod_{i=1}^n (N_R)_{i,\rho(i)}$, which is the product over the sources of the sums over different paths from that source to the desired receiver and thus the sum over the different sets of paths from the sources to the receivers of the product over those paths, we get:
\[\prod_{i=1}^n (N_R)_{i,\rho(i)} = \prod_{i=1}^n 
  \left(\sum_{\pi_{i,\rho(i)}:\text{ a path}} \left(\prod_{(u,v) \in
    \pi_{i,\rho(i)}} \alpha_{(u,v)} \right)\right)= \sum_{\Pi_{R,\sigma} \in \Psi_R
    : \sigma = \rho} \left(\prod_{(u,v) \in \Pi_{R,\sigma}}
    \alpha_{(u,v)}\right)\]
so by the uniqueness of $\sigma$ for a given $\Pi_{R,\sigma}$, we have:
\[\det(N_R) = \sum_{\rho\in\mathcal{S}_n} \sign(\rho) \sum_{\Pi_{R,\sigma} \in \Psi_R: \sigma = \rho} \left(\prod_{(u,v) \in \Pi_{R,\sigma}} \alpha_{(u,v)} \right)= \sum_{\Pi_{R,\sigma} \in \Psi_R} \sign(\sigma) \prod_{(u,v) \in \Pi_{R,\sigma}} \alpha_{(u,v)}\]
Now the only difference between our current expression for $\det(N_R)$ and the desired expression is that the set of paths $\Pi_{R,\sigma}$ for the determinant might not be vertex disjoint. But as a result of the matching in Lemma~\ref{lem:detmatch}, we have that 
\[\sum_{\Pi_{R,\sigma} \in \Psi_R\backslash\Phi_R} \sign(\sigma) \prod_{(u,v) \in \Pi_{R,\sigma}} \alpha_{(u,v)} = \sum_{\{\Pi_{R,\sigma},\mu(\Pi_{R,\sigma})\} \subseteq \Psi_R\backslash\Phi_R} 0 = 0\]
making
\[\det(N_R) = \sum_{\Pi_{R,\sigma} \in \Phi_R} \sign(\sigma) \prod_{(u,v) \in \Pi_{R,\sigma}} \alpha_{(u,v)}\]
as desired. 
\end{proof}

\begin{corollary}\label{c:terms}
The number of terms in $\det(N_R)$ is the number of sets of vertex disjoint paths from $S_1,\dots,S_n$ to $R^{(1)},\dots,R^{(n)}$. 
\end{corollary}
This follows from Proposition~\ref{prop:detterms}.

\begin{corollary}
For a receiver placement $V_R$, the $\alpha_{(u,v)}$-degree of
$\det(N_R)$ has degree at most 1.
\end{corollary}
This follows by noting that since the paths are vertex-disjoint, any edge can be traversed at most once among a set of paths. Therefore the corresponding variable can only appear once in a monomial corresponding to some path.


\section{Triangular Semilattice Network}

In this section, we introduce and discuss properties of the triangular semilattice network, a code graph with a structure that visually resembles an inverted equilateral triangle. We then seek to add receiver placements to require a greater minimum field size.

\begin{defi}\label{def:triangular}
Let a triangular semilattice code graph \net{n} of length $n$ for $n
\in \mathbb{N}\setminus\{0\}$ be a code graph with its underlying
directed acyclic graph given by the set:
\begin{itemize}
\item $V = \{(x,y) \in \Z^2 \mid x,y \geq 0, x + y < n\}$ is the
  vertex set;
\item $E = \{((x+1,y),(x,y)): 0 \leq x+y < n-1\} \cup
  \{((x,y+1),(x,y)), 0 \leq x+y < n - 1\}$ is the set of edges.
\end{itemize}

For $1 \leq i \leq n$, we call the set of vertices
$\{(a,b) \mid a+b = n-i\}$ the $i^{th}$ level where the $1^{st}$ level
is denoted the top level and the $n^{th}$ level is denoted the bottom
level. We enumerate the vertices in increasing order of level and then increasing
order of the $x$ coordinate within the level.

We may refer to the triangular semilattice network of length $n$ as
any network with associated code graph \net{n}.
\end{defi}

Figure~\ref{fig:lattice3} shows \net{3} without receiver
labels but the enumeration of the vertices. Later in this work, we will often identify vertices with the value in this enumeration.

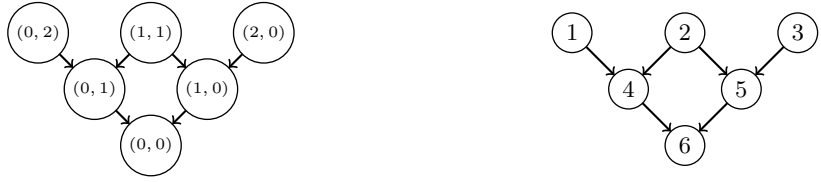
\begin{figure}[htbp]
\centering\begin{subfigure}{.43\textwidth}
\begin{center}
\begin{tikzpicture}
\tikzstyle{every node}=[font=\tiny]
\renewcommand*{\VertexInterMinSize}{20pt}
\SetGraphUnit{.75}
\Vertex[L={$(0,2)$}]{1} \SOEA[L={$(0,1)$}](1){4} \SOEA[L={$(0,0)$}](4){6} \NOEA[L={$(1,1)$}](4){2} \SOEA[L={$(1,0)$}](2){5} \NOEA[L={$(2,0)$}](5){3}

\SetUpEdge[style={->}]
\Edges(1,4,6) \Edges(2,4) \Edges(2,5) \Edges(3,5,6)
\end{tikzpicture}
\end{center}
\caption{Definition of \net{3}.}
\end{subfigure}\begin{subfigure}{.43\textwidth}
\begin{center}
\begin{tikzpicture}
\tikzstyle{every node}=[font=\footnotesize]
\renewcommand*{\VertexInterMinSize}{15pt}
\SetGraphUnit{.75}
\Vertex[L={$1$}]{1} \SOEA[L={$4$}](1){4} \SOEA[L={$6$}](4){6} \NOEA[L={$2$}](4){2} \SOEA[L={$5$}](2){5} \NOEA[L={$3$}](5){3}

\SetUpEdge[style={->}]
\Edges(1,4,6) \Edges(2,4) \Edges(2,5) \Edges(3,5,6)
\end{tikzpicture}
\end{center}
\caption{Enumeration of the vertices.}
\end{subfigure}
\caption{Representation of a triangular semilattice code graph \net{3}
with vertex enumeration.}
\label{fig:lattice3}
\end{figure}

\begin{defi}\label{d:side-receivers}
Let the left-side refer to the $n$ vertices in the \net{n} with $x$-coordinate equal to 0. Similarly the right-side refers to the $n$ vertices with $y$-coordinate equal to 0. We collectively refer to these as the sides.
\end{defi}
Note that embedding \net{n} as above, the left side corresponds with vertices without left children and the right side corresponds with vertices without right children.



\subsection{Valid Receiver Placements}

We introduce some more definitions and lemmas to help us prove the
characterization of valid receiver placements, meaning labeled
vertices distributed such that there is a choice of disjoint paths
between sources and labeled vertices. 

\begin{defi}
A $k$-triangle in a triangular semilattice network \net{n} is a
subgraph isomorphic as a directed graph to a triangular semilattice
network \net{k}. We call $k$ the length of a $k$-triangle. 
\end{defi}
We will drop $k$ if the length of the triangle is clear from the context.
Note that length can also be defined via the length of the longest path between any two vertices in the triangle (also considering number of vertices for length) or the number of vertices along the top of the triangle.

\begin{defi}
  Given a receiver placement of $R$, a $k$-triangle is overcrowded if there are at least $k+1$ labels
  among its vertices. It is crowded if there are exactly $k$ labels. A
  $k$-triangle is distributed if no triangle contained in it is
  overcrowded.
\end{defi}

\begin{remark}
  It is insufficient to just consider $(n-1)$-triangles for the
  distributed property. Consider the network in
  Figure~\ref{fig:bigcrowd}, where the receiver labeled vertices are shown in
  gray. Note that there are 3 labels in a 2-triangle, making it
  not distributed but there are not 4 labels in a 3-triangle.

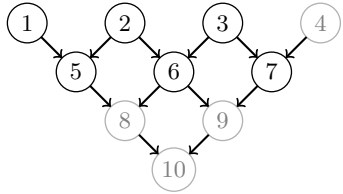
\begin{figure}[htbp]
\begin{center}
\begin{tikzpicture}
\tikzstyle{every node}=[font=\footnotesize]
\renewcommand*{\VertexInterMinSize}{15pt}
\SetGraphUnit{.65}
\Vertex{1} \SOEA(1){5} {\renewcommand{\VertexLineColor}{black!30} \renewcommand{\VertexTextColor}{black!50}\SOEA(5){8} \SOEA(8){10} \NOEA(10){9}} \NOEA(9){7} \NOWE(7){3} {\renewcommand{\VertexLineColor}{black!30} \renewcommand{\VertexTextColor}{black!50} \NOEA(7){4}} \SOWE(3){6} \NOWE(6){2}

\SetUpEdge[style={->}]
\Edges(1,5,8,10) \Edges(2,6,9) \Edges(3,7) \Edges(4,7,9,10) \Edges(3,6,8) \Edges(2,5)
\end{tikzpicture}
\end{center}
\caption{No overcrowded $(n-1)$-triangle but an overcrowded 2-triangle}
\label{fig:bigcrowd}
\end{figure}
\end{remark}

\begin{defi}
We say that two vertices $a$ and $b$ are consecutive if they share a child. A sequence $a_1,\dots,a_k$ of distinct vertices has consecutive vertices if $a_i$ and $a_{i+1}$ are consecutive for every $i=1,\dots,k-1$. A vertex $c$ is between $a$ and $b$ if there is a sequence of consecutive vertices with extremals $a$ and $b$ containing $c$.
\end{defi}
Intuitively, consecutive vertices are ``next to" each other on the same level of the network.

\begin{defi}
For two distinct vertices $a,b$ on the same level, we say that $a$ is to the left of $b$ (equivalently that $b$ is to the right of $a$) if its value in the enumeration is less than that of $b$.
\end{defi}

\begin{defi}
For a vertex $a$ to the left of some vertex $b$,
we say some vertex $c$ is trapped between $a$ and $b$ if the vertex is
in the next level and it is between $a$'s right child and
$b$'s left child.
\end{defi}
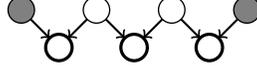
\begin{figure}[htbp]
\begin{center}
\begin{tikzpicture}
\SetGraphUnit{1}
\tikzstyle{every node}=[font=\footnotesize]
\renewcommand*{\VertexInterMinSize}{10pt}
\SetVertexNoLabel
{\renewcommand{\VertexLightFillColor}{black!50}\Vertex{A} \Vertex[x=3,y=0]{D}} \EA(A){B} \EA(B){C} {\renewcommand{\VertexLineWidth}{1.2pt}\Vertices[x=.5,y=-.5]{line}{E,F,G}}
\SetUpEdge[style={->}]
\Edge(A)(E) \Edge(B)(E) \Edge(B)(F) \Edge(C)(F) \Edge(C)(G) \Edge(D)(G)
\end{tikzpicture}
\caption{The thickly-outlined vertices are trapped between the two filled-in vertices}
\label{fig:trapped}
\end{center}
\end{figure}

\begin{defi}
\label{def:extend}
The extension of a $k$-triangle is the $(k+1)$-triangle containing the original $k$-triangle and all parents of the vertices in the $k$-triangle.
\end{defi}

\begin{lem}
\label{lem:crowd}
Let \net{n} be an distributed triangular semilattice network and a sequence of consecutive vertices where each vertex is contained in a crowded triangle. Then, there is a crowded triangle containing all vertices in this sequence.
\end{lem}
\begin{proof}
We induct on the length of the sequence. If there is just one such vertex, we are done.

On two consecutive vertices $x$ and $y$, we have a crowded $k$-triangle corresponding to $x$ which may intersect a crowded $l$-triangle corresponding to $y$ (where $k,l$ are some lengths). Note that if the intersection has length $i \geq 0$, it has at most $i$ labels or we have a contradiction. In that case, consider the triangle of length $k+l-i$ containing the two crowded triangles; note that it contains at least the labels in the $k$-triangle and $l$-triangle, which by Inclusion/Exclusion, have at least $k+l-i$ labels combined. By assumption, a $(k+l-i)$-triangle must have at most $k+l-i$ labels, so we have equality, thus forming a crowded triangle.

Now for our inductive step, assume the result for $m\geq 2$ and consider $m+1$ consecutive vertices contained in crowded triangles. By the inductive hypothesis, we have some crowded $l$-triangle containing the first $m$ vertices. We can then apply the case for two vertices to the $m^{\text{th}}$ vertex (with the crowded $l$-triangle) and the $m+1^{\text{st}}$ vertex (with some crowded $k$-triangle) to get some crowded $j$-triangle containing all $m+1$ vertices (where $j,k,l$ are some lengths).
\end{proof}

\begin{lem}
\label{lem:trapmatch}
Let \net{n} be an distributed triangular semilattice network with
$t>1$ labels in the top level. Then, there are $t-1$ unlabeled vertices in the second level such that upon labeling them, the bottom $(n-1)$-triangle is distributed.
\end{lem}

\begin{proof}
Let $L$ be the leftmost labeled vertex in the top level. Note that it suffices to show that iteratively, for every top level labeled vertex $v \neq L$,
we can label a previously unlabeled vertex trapped by $u$, the rightmost labeled vertex to the left of $v$, and $v$ such that the bottom $(n-1)$-triangle is distributed.

We prove the claim by contraposition: assume that at some point, there exists a labeled vertex $v \neq L$ in the top level such that we create an overcrowded triangle in the bottom $(n-1)$-triangle for every such labeling. Then, we show that there was originally an overcrowded triangle in the network. In particular, we claim that if every labeling creates an overcrowded triangle, every vertex trapped by $v$ and the previous labeled vertex $u$ is in some crowded triangle. Each of the labeled trapped vertices forms a crowded 1-triangle. Moreover, by assumption, upon labeling each of the unlabeled trapped vertices, it is in a $k$-triangle with $\geq k+1$ labeled vertices. Without that added label, we thus have $\geq k$ labeled vertices in a $k$-triangle. If we have more than $k$ labels in this $k$-triangle, we arrive at a contradiction, otherwise, we have a crowded triangle. We can then apply Lemma~\ref{lem:crowd} to get a crowded $l$-triangle containing all of the trapped vertices. From there, we can extend the triangle to the first level to include $u$ and $v$ as in Definition~\ref{def:extend}, getting $l+2$ labels in an $(l+1)$-triangle in the original graph.
\end{proof}

\begin{thrm}
\label{thm:valid}
  Given a triangular semilattice network \net{n}, a labeling $V_R$ of
  $n$ vertices corresponding to some receiver $R$ is valid, meaning
  that there are vertex-disjoint paths to the vertices labeled by
  $V_R$ from the sources, if and only if the network is distributed.
\end{thrm}

\begin{proof}
We first show the forward direction. Fix a valid receiver placement and a triangle of length $k$. Consider the set $S_k$ of the vertices corresponding to the labels in the triangle. Note that the $\mincut$ from the sources to the set $S_k$ is at most $k$, since the top level of the triangle is a cut of size $k$. As such, by Menger's Theorem, there are at most $k$ vertex-disjoint paths to the set $S_k$, and thus, at most $k$ labels in the triangle. 

We now show the other direction by induction on $n$. The base cases of $n=1,2$ are trivial. Now assume the result for $n \geq 2$. Consider a triangular semilattice network \net{n+1} and a receiver placement satisfying the desired property. As there are at most $n$ labels in the bottom triangle of length $n$, there must be at least one label in the top level. We call the leftmost label $L$ and match the remaining $n$ vertices in the top level with the next level as follows.

If there is only one label in the top level, we can iteratively
match/biject all vertices in the first level, from left to right, to
the leftmost unmatched vertex in the next level---in particular, we
match the vertices to the left of $L$ with their right child and those
to the right of $L$ with their left child. Applying the inductive
hypothesis to the bottom $n$-triangle, we can extend the $n$
vertex-disjoint paths from the second level to the receivers to begin
at the sources via the matching. With $\{L\}$, we then have our $n+1$
vertex-disjoint paths to the labels.

Otherwise there are at least two labels in the top level. 
By Lemma~\ref{lem:trapmatch}, we have a matching of the labeled vertices in the top level to some trapped vertices in the next level. Note that if we enumerate the top level's vertices as $a_1,\dots,a_{n+1}$ and the second level's vertices as $b_1,\dots,b_n$, a vertex $a_i$ has children $b_{i-1}, b_i$ if $i-1,i \in [n]$. Now, we match each remaining unlabeled vertex in the top level with an unmatched child as follows:
\begin{itemize}
\item We can match any consecutive vertices $a_1,\dots,a_m$ up to $L$ (exclusive) by matching $a_i$ with $b_i$ for $i=1,\dots,m$. None of those $b_i$ have been matched as they are not trapped by any two labeled vertices.
\item We can match any consecutive vertices $a_t,\dots,a_{n+1}$ after the rightmost labeled vertex in the top level by matching $a_{i+1}$ with $b_i$ for $i=t-1,\dots,n$. Again we note that none of these $b_i$ are trapped by any two labeled vertices.
\item For the unlabeled vertices $a_r,\dots,a_s$ between two labeled vertices $u$ and $v$ in the top level, we match these to $\{b_{r-1},\dots,b_s\}\backslash\{b_p\}$ where $b_p$ is the vertex matched to $v$. For $1<r\leq i\leq p$, we match $a_i$ with $b_{i-1}$ and for $p < i \leq s$, we match $a_i$ with $b_i$.
\end{itemize}
Note that this process creates an bijection between vertices. Within a
section (between the trapped vertices or at the ends), the process is
clearly injective. Across the consecutive sections, we reach a label
at position $a_j$ where the furthest right vertex the left
section matches to is $b_{j-1}$ (and sections further left match
to vertices further left) and the furthest left vertex the right
section matches to is $b_j$.

Finally, by our inductive hypothesis, we have vertex-disjoint paths from the sources/top level of the bottom $n$-triangle to the labels originally there and those added by Lemma~\ref{lem:trapmatch}. The set of vertex-disjoint paths in the original $(n+1)$-triangle is then as follows. Every label in the top level is just a path with a single vertex. For every other label in a lower level, we extend the path found in the bottom $n$-triangle via the matching with the unlabeled sources that we just found. This is vertex disjoint as there are no intersections in the top level and the paths when restricted to the bottom $n$-triangle are either empty or are as found in the inductive hypothesis.
\end{proof}

We can further locate some receiver placements with well-understood
determinants. Previously we denoted transfer coefficients using $\alpha_{(u,v)}$
where $(u,v)\in E$. Henceforth we use $\alpha^{(i)}_j$ for the
transfer coefficients of the triangular semilattice network \net{n}
for $i\in [\abs{\netm{n-1}}]$, where $\abs{\netm{n-1}}$ is the number of vertices in \net{n-1} and thus the bottom $(n-1)$-triangle of \net{n}, and $j \in [2]$. Here, $\alpha^{(i)}_1$ is the transfer
coefficient of the edge between vertex $(i+n)$ and its left parent and
$\alpha^{(i)}_2$ is the one between $(i+n)$ and its right parent. 

\begin{prop}
\label{prop:oneterm}
Let $V_R$ be a receiver placement in a triangular semilattice network
\net{n} consisting of exactly one label per level where each label is
along the sides of the network, and let $V_{R'}$ be the reflected receiver
placement, meaning that its labels are the remaining side labels
together with the bottom one. Then 
\[\det(N_R)\det(N_{R'})=\pm\prod_{i \in [\abs{\netm{n-1}}], j\in [2]}\alpha_j^{(i)}\]
\end{prop}

\begin{proof}
We prove this by induction on the length $n$ of the triangular
semilattice network \net{n}. This is trivial in the case of \net{1}, as there are no variables. In the case of \net{2}, we either take the right source and the bottom vertex---for a determinant of $\alpha_2^{(1)}$---or the left source and the bottom vertex---for a determinant of $\alpha_1^{(1)}$, and we have the product is then $\alpha_1^{(1)}\alpha_2^{(1)}$, as desired.

Now consider the triangular semilattice network \net{n+1} for $n \in \N, n \geq 2$ where we fix a receiver placement such that we have a label in each level along the sides. Let $N_R$ be the submatrix corresponding to this receiver placement. Consider:
\[L = \begin{pmatrix}1&\alpha^{(1)}_1&0&\cdots&0\\0&\alpha^{(1)}_2&\alpha^{(2)}_1&\cdots&0\\\vdots&\vdots&\ddots&\vdots\\0&0&0&\cdots&\alpha^{(n)}_2\end{pmatrix}\qquad\text{or}\qquad L_{i,j} = \begin{cases}1&\text{if }i=j=1\\\alpha_1^{(j-1)}&\text{if }i+1=j\geq2\\\alpha_2^{(j)}&\text{if }i=j\geq2\\0&\text{otherwise}\end{cases}\]
and
\[T = \begin{pmatrix}0&\alpha^{(1)}_1&0&\cdots&0\\0&\alpha^{(1)}_2&\alpha^{(2)}_1&\cdots&0\\\vdots&\vdots&\ddots&\vdots\\1&0&0&\cdots&\alpha^{(n)}_2\end{pmatrix}\qquad\text{or}\qquad T_{i,j} = \begin{cases}1&\text{if }i=n \text{ and }j=1\\\alpha_1^{(j-1)}&\text{if }i+1=j\geq2\\\alpha_2^{(j)}&\text{if }i=j\geq2\\0&\text{otherwise}\end{cases}\]
Extending to the field of fractions $\F_q(\alpha^{(i)}_j\mid i \in
[\abs{\netm{n}}],\  j \in [2])$, note that $L^{-1}$ corresponds to
the basis change taking the leftmost label and the 2$^{\text{nd}}$
level and $T^{-1}$ corresponds to the basis change taking the
rightmost label and the 2$^{\text{nd}}$ level. Further note that
$\det(L) = \prod_{i=1}^{n} \alpha^{(i)}_2$ and $\det(T) =
\pm\prod_{i=1}^n \alpha^{(i)}_1$. To calculate $\det(N_R)$, it suffices
to calculate $\det(LL^{-1}N_R)=\det(L)\det(L^{-1}N_R)$ or
$\det(TT^{-1}N_R)=\det(T)\det(T^{-1}N_R)$. Now after the basis change
(using $L$ if we picked the top left label and $T$ if we picked the
top right label), the label's structure of the bottom $n$-triangle is
identical to that of a triangular semilattice network \net{n}.

Further note that the basis-changed matrix $\bar{N}_R = L^{-1}N_R$ or $T^{-1}N_R$ is in the block matrix form of
\[\bar{N}_R=\begin{pmatrix}1&0\\0&\bar{N}_R'\end{pmatrix}\]
where $\bar{N}_R'$ is the matrix corresponding to the bottom $n$
labels in \net{n}. Expanding by minors, we have
$\det(\bar{N}_R) = \det(\bar{N}_R')$. By inductive
hypothesis we have that $\det(\bar{N}_R')$ is a monomial where the
product of this determinant and that corresponding to the reflection
of the bottom $n$ labels is a monomial with all transfer
coefficients in \net{n}.
As switching between the leftmost top label and the rightmost top label swaps between $L$ and $T$, combining this with the bottom $n$-triangle for the original determinants, we get the desired result.
\end{proof}

As a consequence we obtain that a receiver placement $V_R$ for a
triangular semilattice network \net{n} defined as in Proposition
\ref{prop:oneterm} is a valid receiver for any choice of triangular semilattice network of length $n$
$n$ and there exists an $\F_q$-labeling with nonzero transfer coefficients for any finite field $\F_q$. As such, for the rest of the paper we
consider the triangular semilattice network \net{n} to be equipped
with two receivers: the left-side and the right-side receiver, meaning
the receivers with placements $\{(0,n-1),\dots,(0,0)\}$ and
$\{(n-1,0),\dots,(0,0)\}$ respectively as defined in Definition
\ref{d:side-receivers}.

\subsection{Invariance Under Symmetries of Receiver Placements}

In this section we study properties of minors of $\F_q$-labelings from receiver
placements. We will show that the property of having a
$\F_q$-labeling for a receiver placement implies the existence of an
$\F_q$-labeling for any receiver placement that is obtained from the
original from either rotation or reflection with respect to the
underlining graph of the network.

\begin{defi}
\label{def:tsnrotate}
Let \net{n} defined as in Definition \ref{def:triangular}. Then, the
maps $\rho:V\rightarrow V$ defined as
$\rho(x,y)=(n-1-x-y,x)$ and the map $\sigma: V\rightarrow V$ defined
as $\sigma(x,y)=(y,x)$ are a bijections of the set of vertices with
$\rho^3=id$ and $\sigma^2=id$ respectively. 
\end{defi}

Roughly speaking, $\rho$ represents a counterclockwise rotation of the
vertices whereas $\sigma$ represents a reflection. These two maps can
be naturally extended to subsets of vertices. We are going to use
these maps prevalently on receivers placements, meaning that the directed
structure of the network is not going to change.  Let $V_R=\{v\in V \mid R\in L_v\}$ be a receiver
placement, then $V_{\rho(R)}:=\{\rho(v)\in V \mid R\in L_v\}$ and
$V_{\sigma(R)}:=\{\sigma(v)\in V \mid R\in L_v\}$ are two others receiver
placements. Figure \ref{f:rot/ref}
provide examples for the 3-semilattice network.

\begin{figure}[htbp]
\begin{subfigure}{.33\textwidth}
\begin{center}
\begin{tikzpicture}
\tikzstyle{every node}=[font=\footnotesize]
\renewcommand*{\VertexInterMinSize}{15pt}
\SetGraphUnit{.65}
\Vertex[L={$R$}]{1} \SOEA[L={$R$}](1){4} \SOEA[NoLabel](4){6} \NOEA[NoLabel](4){2} \SOEA[L={$R$}](2){5} \NOEA[NoLabel](5){3}

\SetUpEdge[style={->}]
\Edges(1,4,6) \Edges(2,4) \Edges(2,5) \Edges(3,5,6)
\end{tikzpicture}
\end{center}
\caption{$V_R$ receiver placement.}
\end{subfigure}\begin{subfigure}{.33\textwidth}
\begin{center}
\begin{tikzpicture}
\tikzstyle{every node}=[font=\footnotesize]
\renewcommand*{\VertexInterMinSize}{15pt}
\SetGraphUnit{.65}
\Vertex[NoLabel]{1} \SOEA[NoLabel](1){4} \SOEA[L={$R$}](4){6} \NOEA[L={$R$}](4){2} \SOEA[L={$R$}](2){5} \NOEA[NoLabel](5){3}

\SetUpEdge[style={->}]
\Edges(1,4,6) \Edges(2,4) \Edges(2,5) \Edges(3,5,6)
\end{tikzpicture}
\end{center}
\caption{$V_{\rho(R)}$ receiver placement.}
\end{subfigure}\begin{subfigure}{.33\textwidth}
\begin{center}
\begin{tikzpicture}
\tikzstyle{every node}=[font=\footnotesize]
\renewcommand*{\VertexInterMinSize}{15pt}
\SetGraphUnit{.65}
\Vertex[NoLabel]{1} \SOEA[L={$R$}](1){4} \SOEA[NoLabel](4){6} \NOEA[NoLabel](4){2} \SOEA[L={$R$}](2){5} \NOEA[L={$R$}](5){3}

\SetUpEdge[style={->}]
\Edges(1,4,6) \Edges(2,4) \Edges(2,5) \Edges(3,5,6)
\end{tikzpicture}
\end{center}
\caption{$V_{\sigma(R)}$ receiver placement.}
\end{subfigure}
\caption{Receiver placements of the 3-semilattice.}
\label{f:rot/ref}
\end{figure}
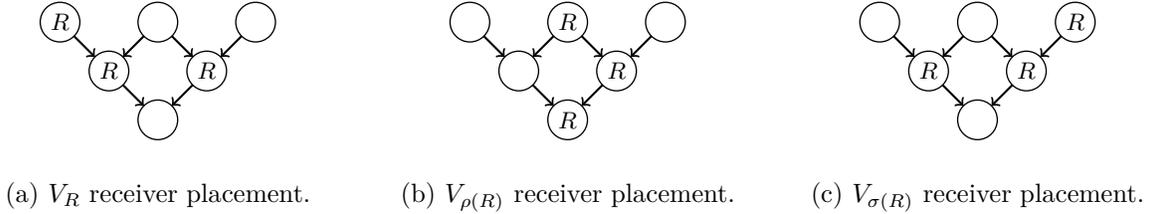

\begin{thrm}
  Let $V_R$ be a receiver placement for a triangular semilattice network
  \net{n}. The following hold:
  \begin{enumerate}
  \item $V_{\rho(R)}$ and $V_{\sigma(R)}$ are valid if and only if $V_R$
    is valid.
  \item If \net{n} is equipped with the side receivers, there exists
    an $\F_q$-labeling for \net{n} with valid receiver placements
    $V_{\rho(R)}$ or $V_{\sigma(R)}$ if and only if there exists an
    $\F_q$-labeling for \net{n} with the valid receiver placement $V_R$.
  \end{enumerate}
\end{thrm}

\begin{proof}
  By Theorem~\ref{thm:valid}, The receiver placement $V_R$ is valid if the triangular semilattice
  network \net{n} with the labels in $V_R$ is distributed. It
  is evident that being distributed is a property of the labeled
  network which is preserved by rotation or reflection of the
  labels. So it holds that $V_{\rho(R)}$ and $V_{\sigma(R)}$ are valid
  if and only if $V_R$ is valid.

  Let $N$ be a $\F_q$-labeling of the triangular semilattice network
  \net{n} with side receivers.  Let $V_\S$, $V_\ell$ and $V_r$ in $V$
  refer to the placements of the
  sources, the left receiver and the right receiver respectively. It
  holds that: 
  \begin{align}
    &V_{\rho(\ell)}=V_r \textrm{ and } V_{\rho(r)}=V_\S \label{eq:1} \\
    &V_{\sigma(\ell)}=V_r \textrm{ and } V_{\sigma(\S)}=V_\S.\label{eq:2}
  \end{align}
Let $V_R$ be a valid
  receiver placement and let $N$ be a $\F_q$-labeling. Let $N_v$ denote the column of $N$ corresponding to vector $v\in V$ and
  $N_{T}$ denote the submatrix of $N$ with columns indexed by $T\subseteq V$.
  Let $N^\rho$ be the matrix defined by the relation
  $N^\rho_v:=N_{\rho^{-1}(v)}$. Up to a multiplication of an
  invertible $|\S|\times |\S|$ matrix, $N^\rho$ is a $\F_q$-labeling
  of \net{n} with side receivers and receiver placement
  $V_{\rho(R)}$. In fact, by Equation \eqref{eq:1}, $N^\rho_\S$, $N^\rho_\ell$, $N^\rho_r$ and
  $N^\rho_{V_{\rho(R)}}$ are invertible since, up to reordering of the
  columns, they correspond to matrices $N_r$, $N_\S$, $N_\ell$ and
  $N_{V_R}$. A similar reasoning works for the reflection map $\sigma$.
  
\end{proof}

\section{Complete Study of Triangular Semilattice Network up to 4 Sources}
In this section, we will demonstrate various properties relating to the receiver placements and minimum field sizes required to solve the $\F_q$-labeling conditions for the triangle semilattice network on small lengths.
\subsection{The Triangular Semilattice Networks \net{2} and \net{3}}
The 2-semilattice has three different valid receiver placements and is trivially solvable over $\mathbb{F}_2$. Note that it is the code graph for the butterfly network.

The 3-semilattice has 17 different valid receiver placements. Excluding the receiver placement corresponding to the three corner nodes, all valid receiver placements have one term in its associated minor. Therefore, any choice of receiver placements that does not include that receiver placement is solvable over $\mathbb{F}_2$ by assigning all of the variables a value of $1$. When receiver placements are chosen to include those receiver placements along the left-side, along the right-side, and corresponding to the three corner nodes, $\mathbb{F}_2$ will cause one associated minor to equal zero, so the minimum field size over which the network is solvable is $\mathbb{F}_3$ 

\subsection{The Triangular Semilattice Network \net{4}}

The triangular semilattice network \net{4} has 150 possible receiver
placements. Through exhaustion (see appendix), we know that
$\mathbb{F}_5$ is sufficient for \net{4} to be solvable when all 150
receiver placements are considered. We consider \net{4} together with the side
receivers and we find the solvability of the network by increasing its
receivers.

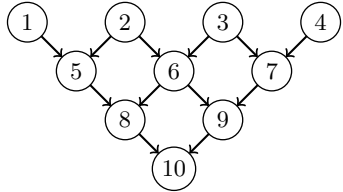
\begin{figure}[htbp]
\begin{center}
\begin{tikzpicture}[scale=.65]
\tikzstyle{every node}=[font=\footnotesize]
\renewcommand*{\VertexInterMinSize}{15pt}
\SetGraphUnit{1}
\Vertices[dir=\NOEA]{line}{10,9,7,4} \Vertices[x=-1,y=1,dir=\NOEA]{line}{8,6,3} \Vertices[x=-2,y=2,dir=\NOEA]{line}{5,2} \Vertex[x=-3,y=3]{1}
\SetUpEdge[style={->}]
\Edges(4,7,9,10) \Edges(2,6,9) \Edge(3)(7) \Edges(1,5,8,10) \Edges(2,5) \Edges(3,6,8)
\end{tikzpicture}
\end{center}
\caption{The triangular semilattice network \net{4}}
\end{figure}

\begin{prop}
The semilattice network \net{4} together with any 2 receivers is
solvable over $\F_q$ for $q \leq 3$.
\end{prop}

\begin{proof}
  First recall that having $\{1,5,8,10\}$ and $\{4,7,9,10\}$ as
  receiver placements forces every transfer coefficient of \net{4} to be nonzero
  as shown in Proposition~\ref{prop:oneterm}. Moreover, let $V_R$ be a receiver
  placement, then, from Corollary \ref{c:terms}, $\det(N_R)\in \F_q[\alpha^{(i)}_j\mid i \in
[6],\  j \in [2]]$ is a
  multivariate polynomial with at most three terms. 

  Let $[i,j]$ for $1 \leq i \leq j \leq 3$ represents the number of
  terms of the minors corresponding to the two further receivers where
  we assume $i \leq j$ without loss of generality. We are going to
  prove the theorem by working through the different
  cases. 

\begin{itemize}[leftmargin=.5cm]
\item $[1,1], [1,3], [3,3]$: The minors have all odd numbers of terms and by
  setting all variables to 1 over $\F_2$, the value of every
  minor is then 1.

\item $[1,2], [2,2]$: Set all variables to 1 over $\F_3$, the value of the 1-term minor would be 1 and the value of the 2-term minor(s) would be 2.

\item $[2,3]$: This case is not solvable over $\F_2$ since then the
  2-term minor is 0. We prove that this case is solvable over $\F_3$
  by contradiction; assume that for every evaluation point
  $\mathsf{a}=(\sfa^{(i)}_j\mid i \in [6],\ j \in [2])\in
  \F_3^{12}$ without zero entries at least one of the minors is zero.

  Since all transfer coefficients must be nonzero, without loss of
  generality we can denote the minors as $A+B$ and $C+D+E$ where $A,B$
  and $C,D,E$ are 
  terms with no common factor respectively. In the following, swapping
  a nonzero value $a\in \F_3$ corresponds to taking the value $2a\in \F_3$.
  
\begin{enumerate}[label=\arabic*)]
\item[(i)]\label{item1} Let $\mathsf{a}\in \F_3^{12}$ be such that
  $(A+B)(\sfa)=(C+D+E)(\sfa)=0$ and $\alpha^{(i)}_j$ be a variable in
  $C+D+E$. Define $\sfa'\in \F_3^{12}$ to be equal to $\sfa$ except for
  $\sfa^{(i)}_j$, which is swapped; then
  $(C+D+E)(\sfa') \neq 0$. If the same $\alpha^{(i)}_j$ appears in
  $A+B$ as well, we have $(A+B)(\sfa')\neq 0$, a contradiction. If no
  variable in $C+D+E$ appears in $A+B$, instead define $\sfa'\in\F_3^{12}$
  from $\sfa$ by swapping two values of it corresponding to some
  variable in $A+B$ and to some variable in $C+D+E$ independently to
  again get $(C+D+E)(\sfa')\neq 0, (A+B)(\sfa') \neq 0$, a
  contradiction again.

\item[(ii)] Let instead for all $\sfa\in \F_3^{12}$ exactly one of
  $(A+B)(\sfa)$ and $(C+D+E)(\sfa)$ is 0.

  Note that all variables in $A+B$ must appear in $C+D+E$; assume for
  the sake of contradiction that there is some variable
  $\alpha_j^{(i)}$ which appears in $A+B$ which does not appear in
  $C+D+E$. Then, if $(A+B)(\sfa)=0$ and $(C+D+E)(\sfa)\neq0$, the
  evaluation point $\sfa'\in\F_3^{12}$ defined from $\sfa$ by swapping
  the value of $\sfa_j^{(i)}$ produces $(A+B)(\sfa')\neq0$ and
  $(C+D+E)(\sfa')\neq0$, a contradiction. If
  $(C+D+E)(\sfa)=0,(A+B)(\sfa)\neq0$, then taking $\sfa'\in\F_3^{12}$
  defined from $\sfa$ by swapping the value of $\sfa_j^{(i)}$ produces
  $(A+B)(\sfa')=0,(C+D+E)(\sfa')=0$, again a contradiction for 
  Item (i). This proves that all variables in $A+B$ must appear in $C+D+E$.

  \smallskip Let $\sfa\in \F_3^{12}$ be such that $(A+B)(\sfa)= 0$ and
  $(C+D+E)(\sfa)\neq 0$, then if $\sfa'$ is obtained by $\sfa$ by
  swapping one of the values corresponding to a variable contained in
  $A+B$, then $(A+B)(\sfa')\neq 0$ and $(C+D+E)(\sfa')=0$.

  Without loss of generality we can focus on the case where
  $\sfa\in \F_3^{12}$ is such that $(A+B)(\sfa)\neq 0$ and
  $(C+D+E)(\sfa)=0$. 
  \begin{itemize}
  \item Consider now the case where there exists a variable
    $\alpha_j^{(i)}$ which appears in $C+D+E$ but not in $A+B$ and
    define $\sfa'\in \F_3^{12}$ from $\sfa$ by swapping the value of
    $\sfa_j^{(i)}$. Then, $(A+B)(\sfa') \neq 0$ and
    $(C+D+E)(\sfa')\neq 0$, a contradiction.
  \item Consider instead the case where $A+B$ and $C+D+E$ share the same
    set of variables. Let $\sfa\in \F_3^{12}$ be a root of $C+D+E$.
    As each swap changes whether $A+B$ is nonzero, if $\sfa'\in \F_3^{12}$ is obtained from $\sfa$ by swapping
    the values of two distinct variables
    $\alpha_{j_1}^{(i_1)},\alpha_{j_2}^{(i_2)}$ contained in $C+D+E$, we
    get back to $(C+D+E)(\sfa')=0$.  Indeed, either the distinct
    variables appear in the same terms or they partition the
    terms. 
    Note that two such variables
    $\alpha_{j_1}^{(i_1)},\alpha_{j_2}^{(i_2)}$ partitioning the terms
    exist since we cannot have everything sharing the same terms by
    assumption. Then, we are able to partition all variables as to
    whether they share a term with $\alpha_{j_1}^{(i_1)}$ or $\alpha_{j_2}^{(i_2)}$, so we can represent our
    sum in the form of $C+C+E$. This is impossible as the minor is
    formed by a sum of the product of transfer coefficients over
    different sets of paths while the repetition of $C$ corresponds to
    a repeated set of paths.
  \end{itemize}
\qedhere
\end{enumerate}
\end{itemize}
\end{proof}

We can also characterize some sets of receivers in the 4-semilattice which require a larger field size. 

\begin{prop}
\label{prop:sum3}
There exists a choice of three receivers of the semilattice network
\net{4} which is not solvable over $\F_q$ for $q \leq 3$ but it is over $\F_4$.
\end{prop}
\begin{proof}
  We prove that there is no evaluation point $\sfa\in \F_q^{12}$ 
  without zero entries for $q=2,3$ such that the minors related to receiver placements
  $\{2,5,7,10\},$ $\{2,4,9,10\},\{1,4,5,10\}$ are simultaneously
  nonzero.  It holds that
\begin{align*}
\det(N_{\{2, 5, 7, 10\}})&=\alpha_1^{(1)}\alpha_2^{(2)}\alpha_2^{(3)}\alpha_2^{(4)}\alpha_1^{(6)} +	\alpha_1^{(1)}\alpha_2^{(2)}\alpha_2^{(3)}\alpha_1^{(5)}\alpha_2^{(6)}=A+B\\
\det(N_{\{2, 4, 9, 10\}})&=\alpha_1^{(1)}\alpha_2^{(2)}\alpha_1^{(4)}\alpha_1^{(5)}\alpha_1^{(6)} +	\alpha_1^{(1)}\alpha_1^{(3)}\alpha_1^{(4)}\alpha_2^{(5)}\alpha_1^{(6)}
=C+D\\
\det(N_{\{1, 4, 5, 10\}})&=\alpha_2^{(1)}\alpha_2^{(2)}\alpha_2^{(4)}\alpha_1^{(6)} +	\alpha_2^{(1)}\alpha_2^{(2)}\alpha_1^{(5)}\alpha_2^{(6)} + \alpha_2^{(1)}\alpha_1^{(3)}\alpha_2^{(5)}\alpha_2^{(6)}
\\
&=\alpha_2^{(1)}\frac{(A+B)(C+D)-AD}{\left(\alpha_1^{(1)}\right)^2\alpha_2^{(2)}\alpha_2^{(3)}\alpha_1^{(4)}\alpha_1^{(5)}\alpha_{1}^{(6)}}
\end{align*}
  
It is enough at least one of the three polynomials of the
form $A+B, C+D$ and $(A+B)(C+D)-AD$ evaluate to zero.  Over $\F_2$, note that
$(A+B)(\sfa)=1+1=0$. Over $\F_3$, if either $(A+B)(\sfa)=0$ or
$(C+D)(\sfa)=0$, we are done. Otherwise, if there exists
$\sfa\in\F_3^{12}$ such that $(A+B)(\sfa)\neq 0$ and $(C+D)(\sfa) \neq 0$,
then $A(\sfa)=B(\sfa)$ and $C(\sfa)=D(\sfa)$. It follows that 
\begin{align*}
((A+B)(C+D)-AD)(\sfa) &=
  ((2A)(2D))(\sfa)-(AD)(\sfa)  =(AD-AD)(\sfa)=0\label{eq:4}.
\end{align*}

A solution over $\F_4 = \bigslant{\F_2}{(a^2+a+1)}$ for \net{4} with
receiver placements $\{2,5,7,10\},$ $\{2,4,9,10\}$ and  $\{1,4,5,10\}$ is  
\[\sfa=(1,a+1,a+1,a+1,a,a,a+1,a+1,a,1,a,a)\in \F_4^{12}.\]

\qedhere
\end{proof}

By exhaustive search (see appendix), there exist 324 choices of 3
receiver placements (fixing the sides) that require a minimum field
size of $\mathbb{F}_4$ to be solved. Also through exhaustive
search, we know that any selection of up to 5 receiver
placements is solvable over $\mathbb{F}_4$ or a smaller finite
field.



\begin{prop}
  There exists a choice of six receivers of the semilattice network
  \net{4} which is not solvable over $\F_q$ for $q \leq 4$ but it is
  over $\F_5$.
\end{prop}

\begin{proof}

We prove that there is no evaluation point $\sfa\in \F_q^{12}$ 
  without zero entries for $q\leq 4$ such that the minors related to
  receiver placements \[\{1,2,4,9\},\{1,3,4,8\},\{2,5,7,10\},\{1,4,8,9\},
\{1,4,5,10\},\{1,3,4,10\}\] are simultaneously nonzero. It holds that 
\begin{align*}
\det(N_{\{1, 2, 4, 9\}})&= \alpha_2^{(2)}\alpha_1^{(5)} +\alpha_1^{(3)}\alpha_2^{(5)} = A + B\\
\det(N_{\{1, 3, 4, 8\}})&= \alpha_2^{(1)}\alpha_1^{(4)} + \alpha_1^{(2)}\alpha_2^{(4)} = C + D\\
\det(N_{\{2, 5, 7, 10\}})&= \alpha_1^{(1)}\alpha_2^{(2)}\alpha_2^{(3)}\alpha_2^{(4)}\alpha_1^{(6)} + \alpha_1^{(1)}\alpha_2^{(2)}\alpha_2^{(3)}\alpha_1^{(5)}\alpha_2^{(6)} = E+F\\
\det(N_{\{1, 4, 8, 9\}})&= \alpha_2^{(1)}\alpha_2^{(2)}\alpha_1^{(4)}\alpha_1^{(5)} +\alpha_2^{(1)}\alpha_1^{(3)}\alpha_1^{(4)}\alpha_2^{(5)} + \alpha_1^{(2)}\alpha_1^{(3)}\alpha_2^{(4)}\alpha_2^{(5)} = (A+B)(C+D)-AD\\
\det(N_{\{1, 4, 5, 10\}})&= \alpha_2^{(1)}\alpha_2^{(2)}\alpha_2^{(4)}\alpha_1^{(6)} + \alpha_2^{(1)}\alpha_2^{(2)}\alpha_1^{(5)}\alpha_2^{(6)} + \alpha_2^{(1)}\alpha_1^{(3)}\alpha_2^{(5)}\alpha_2^{(6)}= \alpha_2^{(1)}\tfrac{(A+B)(E+F)-BE}{\alpha_1^{(1)}\alpha_2^{(2)}\alpha_2^{(3)}\alpha_1^{(5)}}\\
\det(N_{\{1, 3, 4, 10\}})&=\alpha_2^{(1)}\alpha_1^{(4)}\alpha_1^{(6)} + \alpha_1^{(2)}\alpha_2^{(4)}\alpha_1^{(6)} +\alpha_1^{(2)}\alpha_1^{(5)}\alpha_2^{(6)} = \tfrac{(C+D)(E+F)-CF}{\alpha_1^{(1)}\alpha_2^{(2)}\alpha_2^{(3)}\alpha_2^{(4)}}
\end{align*}

The cases of $q=2,3$ follow from Proposition~\ref{prop:sum3} by just
considering $A+B,C+D,(A+B)(C+D)-AD$. Let $\F_4 =
\bigslant{\F_2}{(a^2+a+1)}$ and $\sfa\in \F_4^{12}$ be such that
$(A+B)(\sfa)\neq 0$, $(C+D)(\sfa)\neq 0$ and $(D+E)(\sfa)\neq
0$. Since $\sfa$ is not a zero of $A,C,E$, we can normalize the sums 
\[(A+B)(\sfa) = A(\sfa)(1+b'), (C+D)(\sfa)=C(\sfa)(1+d'), (E+F)(\sfa)=E(\sfa)(1+f')\]
where $b',d',f'\in \F_4^*$. It also holds that
\[\frac{((A+B)(C+D)-AD)(\sfa)}{(AC)(\sfa)}=(1+b')(1+d')-d'=1+b'+b'd'\]
\[\frac{((A+B)(E+F)-BE)(\sfa)}{(AE)(\sfa)}=(1+b')(1+f')-b'=1+f'+b'f'\]
\[\frac{((C+D)(E+F)-CF)(\sfa)}{(CE)(\sfa)}=(1+d')(1+f')-f'=1+d'+d'f'\]
 If any of $1+b', 1+d', 1+f'$ are 0, then we are done. Otherwise if
 all of $1+b', 1+d', 1+f'$ are nonzero, then $b',d',f'\in\{a,a+1\}$,
 and by the Pigeonhole Principle we have that two of them are
 equal. Without loss of generality, let $b' = d'$, then note that
 $1+b'+b'd' = 1+b'+(b')^2=0$ by the field equation, which implies that $((A+B)(C+D)-AD)(\sfa)=0$.

A solution over  $\F_5$ for \net{4} with
receiver placements $\{1,2,4,9\}$,
$\{1,3,4,8\}$, $\{2,5,7,10\}$, $\{1,4,8,9\}$,
$\{1,4,5,10\}$ and $\{1,3,4,10\}$ is  
\[\sfa=(1, 4, 3, 1, 1, 4, 4, 1, 4, 3, 3, 2)\in \F_5^{12}.\]
\end{proof}

 We have also found that there exist 8748 choices of 6 receiver
placements that are solvable over minimum field size of $\mathbb{F}_5$.

\subparagraph*{Valid Receiver Placements and Field Sizes' Implementations}

Valid receiver placements for triangular semilattice networks \net{n}
for $n$ up to 9 were calculated based on Theorem
\ref{thm:valid} using Python and SML (see Table \ref{table:valid}).

\begin{table}[htbp]
\centering
\begin{tabular}{c|c|c|c}
Length & Valid     & Invalid   & Total     \\ \hline
1      & 1         & 0         & 1         \\
2      & 3         & 0         & 3         \\
3      & 17        & 3         & 20        \\
4      & 150       & 60        & 210       \\
5      & 1848      & 1155      & 3003      \\
6      & 29636     & 24628     & 54264     \\
7      & 589362    & 594678    & 1184040   \\
8      & 14032452  & 16227888  & 30260340  \\
9      & 389622192 & 496540943 & 886163135 \\
\end{tabular}
\caption{Number of Valid Receiver Placements}\label{table:valid}
\end{table}

To calculate whether a set of receiver placements is solvable in a
given field size, we first calculate the minors corresponding to the
receiver placements and multiply them together to get a polynomial
$f$. As in the proof of the Linear Network Coding Theorem in
\cite{Ksc11}, we have a nonzero solution for all of these minors if
and only if $f$ has a nonzero root. This is also true if and only if
the remainder of $f$ modulo $(x_i^q-x_i \mid i \in [n])$ in $\F_q$ is nonzero
\cite[Proposition 2]{GMT08}.  The largest possible minimum field size
required for any set of receiver placements for \net{4} and \net{5} as
been computed implementing this method on MAGMA \cite{magma}.

\section{Acknowledgments}

The authors are grateful to Clemson University for hosting the
REU at which this work was completed. The REU was made possible by an
NSF Research Training Group (RTG) grant (DMS \#1547399) promoting
Coding Theory, Cryptography, and Number Theory at Clemson.


\begin{thebibliography}{10}

\bibitem{ACLY00}
Rudolf Ahlswede, Ning Cai, S-YR Li, and Raymond~W Yeung.
\newblock Network information flow.
\newblock {\em IEEE Transactions on information theory}, 46(4):1204--1216,
  2000.

\bibitem{anderson}
Sarah~E Anderson, Wael Halbawi, Nathan Kaplan, Hiram~H L{\'o}pez, Felice
  Manganiello, Emina Soljanin, and Judy Walker.
\newblock Representations of the multicast network problem.
\newblock {\em arXiv preprint arXiv:1701.05944}, 2017.

\bibitem{magma}
Wieb Bosma, John Cannon, and Catherine Playoust.
\newblock The {M}agma algebra system. {I}. {T}he user language.
\newblock {\em J. Symbolic Comput.}, 24(3-4):235--265, 1997.
\newblock Computational algebra and number theory (London, 1993).

\bibitem{FS06}
Christina Fragouli and Emina Soljanin.
\newblock Information flow decomposition for network coding.
\newblock {\em IEEE Transactions on Information Theory}, 52(3):829--848, 2006.

\bibitem{GMT08}
Olav Geil, Ryutaroh Matsumoto, and Casper Thomsen.
\newblock On field size and success probability in network coding.
\newblock {\em Lecture Notes in Computer Science}, 5130:157--173, 2008.

\bibitem{jaggi05}
S.~Jaggi, P.~Sanders, P.~A. Chou, M.~Effros, S.~Egner, K.~Jain, and L.~M. G.~M.
  Tolhuizen.
\newblock Polynomial time algorithms for multicast network code construction.
\newblock {\em IEEE Transactions on Information Theory}, 51(6):1973--1982, June
  2005.

\bibitem{KM03}
Ralf Koetter and Muriel M{\'e}dard.
\newblock An algebraic approach to network coding.
\newblock {\em IEEE/ACM Transactions on Networking (TON)}, 11(5):782--795,
  2003.

\bibitem{LYC03}
S-YR Li, Raymond~W Yeung, and Ning Cai.
\newblock Linear network coding.
\newblock {\em IEEE transactions on information theory}, 49(2):371--381, 2003.

\bibitem{Ksc11}
Muriel M{\'e}dard and Alex Sprintson.
\newblock {\em Network coding: Fundamentals and applications}.
\newblock Academic Press, 2011.

\bibitem{SYLL15}
Qifu~Tyler Sun, Xunrui Yin, Zongpeng Li, and Keping Long.
\newblock Multicast network coding and field sizes.
\newblock {\em IEEE Transactions on Information Theory}, 61(11):6182--6191,
  2015.

\end{thebibliography}
\end{document}